\documentclass[12pt]{article}       

\usepackage{amsthm}
\usepackage{amsmath}
\usepackage{graphicx}
\usepackage{amssymb}

\newtheorem{thm}{Theorem}[section]
\newtheorem{lma}[thm]{Lemma}
\newtheorem{cor}[thm]{Corollary}

\newtheorem*{claim}{Claim}

\DeclareMathOperator{\comp}{comp}
\DeclareMathOperator{\splitgraph}{split}
\DeclareMathOperator{\partition}{part}
\DeclareMathOperator{\poss}{poss}

\begin{document}  

\title{Spanning trees and the complexity of flood-filling games}
\date{\today}
\author{Kitty Meeks and Alexander Scott\\
\small{Mathematical Institute, University of Oxford, 24-29 St Giles, Oxford OX1 3LB, UK} \\
\texttt{\small{\{meeks,scott\}@maths.ox.ac.uk}}}
\maketitle

\begin{abstract}
We consider problems related to the combinatorial game (Free-) Flood-It, in which players aim to make a coloured graph monochromatic with the minimum possible number of flooding operations.  We show that the minimum number of moves required to flood any given graph $G$ is equal to the minimum, taken over all spanning trees $T$ of $G$, of the number of moves required to flood $T$.  This result is then applied to give two polynomial-time algorithms for flood-filling problems.  Firstly, we can compute in polynomial time the minimum number of moves required to flood a graph with only a polynomial number of connected subgraphs.  Secondly, given any coloured connected graph and a subset of the vertices of bounded size, the number of moves required to connect this subset can be computed in polynomial time.
\end{abstract}

\section{Introduction}

In this paper we consider several problems related to the one-player combinatorial game (Free-)Flood-It, introduced by Arthur, Clifford, Jalsenius, Montanaro and Sach \cite{arthurFUN}.  The game is played on a coloured graph, and the goal is to make the entire graph monochromatic with as few moves as possible.  A move involves picking a vertex $v$ and a colour $d$, and giving all vertices in the same monochromatic component as $v$ colour $d$.  

When the game is played on a planar graph, it can be regarded as modelling repeated use of the flood-fill tool in Microsoft Paint.  Implementations of the game, played on a square grid, are widely available online, and include a flash game \cite{flash} as well as popular smartphone apps \cite{iphoneapp,androidapp}.  Mad Virus \cite{madvirus} is a version of the same game played on a hexagonal grid, and the Honey Bee Game \cite{honeybee} is a two player variant also played on a hexagonal grid, which has been analysed by Fleischer and Woeginger \cite{fleischer10}.

For any coloured graph, we define the following problems.
\begin{itemize}
\item \textsc{Free-Flood-It} is the problem of determining the minimum number of moves required to flood the coloured graph.  The number of colours may be unbounded.
\item $c$-\textsc{Free-Flood-It} is the variant of \textsc{Free-Flood-It} in which only colours from some fixed set of size $c$ are used.
\end{itemize}
A related problem which naturally arises when considering algorithms for Flood-It is to consider the number of moves required to connect a given set of vertices.
\begin{itemize}
\item $k$-\textsc{Linking-Flood-It} is the problem, given a subset $U$ of at most $k$ vertices, of determining the minimum number of moves required to create a monochromatic component containing $U$.  The number of colours may be unbounded.
\end{itemize}
There is another variant of the game in which all moves must be played at some fixed root vertex; we refer to the problem of determining the minimum number of moves required to flood the board in this case as \textsc{Fixed-Flood-It}.\footnote{\textsc{Fixed-Flood-It} is often referred to as simply \textsc{Flood-It}, but we use the longer name to avoid confusion with the free version.}

In \cite{arthurFUN}, Arthur, Clifford, Jalsenius, Montanaro and Sach proved that $c$-\textsc{Free-Flood-It} is NP-hard in the case of an $n \times n$ grid, for every $c \geq 3$, and that this result also holds for the fixed variant.  Lagoutte, Noual and Thierry \cite{lagoutte,lagoutte11} showed that the same result holds when the game is played instead on a hexagonal grid, as in Mad Virus or a one-player version of the Honey Bee Game respectively.  Fleischer and Woeginger \cite{fleischer10} proved that $c$-\textsc{Fixed-Flood-It} remains NP-hard when restricted to trees, for every $c \geq 4$,\footnote{Note that this proof does in fact require four colours, not three as stated in a previous version of \cite{fleischer10}.} and Fukui, Nakanishi, Uehara, Uno and Uno \cite{fukui} demonstrated that this result can be extended to show the hardness of $c$-\textsc{Free-Flood-It} in the same setting.

A few positive results are known, however.  2-\textsc{Free-Flood-It} is solvable in polynomial time on arbitrary graphs, a result shown independently by Clifford et.~al.~\cite{clifford}, Lagoutte \cite{lagoutte} and Meeks and Scott \cite{general}.  It is also known that \textsc{Free-Flood-It} is solvable in polynomial time on paths \cite{clifford,general,fukui} and cycles \cite{fukui}.  Although $c$-\textsc{Free-Flood-It} is NP-hard on rectangular $3 \times n$ boards for any $c \geq 4$ \cite{general}, $c$-\textsc{Free-Flood-It} is fixed parameter tractable with parameter $c$ when restricted to $2 \times n$ boards (Meeks and Scott \cite{2xn}), and the fixed variant can be solved in linear time in this situation \cite{clifford}.  Meeks and Scott \cite{general} also show that 2-\textsc{Linking Flood-It} can be solved in polynomial time for arbitrary graphs, even when the number of colours is unbounded.

In this paper we give some more general tractability results, which do not require the number of colours to be bounded.  Our first such result is that \textsc{Free-Flood-It} can be solved in polynomial time on the class of graphs which have only a polynomial number of connected subgraphs.  This class includes a number of interesting families of graphs, and the result implies a conjecture from \cite{general} that the problem can be solved in polynomial time on subdivisions of any fixed graph.  We also consider the fixed variant of the game on the same class of graphs, and show that the minimum number of moves can again be computed in polynomial time in this case.  This substantially extends results of Lagoutte \cite{lagoutte}, and Fukui, Nakanishi, Uehara, Uno and Uno \cite{fukui}, who showed that the fixed and free variants respectively are polynomially solvable on cycles.

We then go on to consider $k$-\textsc{Linking-Flood-It}.  We prove that, for any fixed $k$, it is possible to solve $k$-\textsc{Linking-Flood-It} in polynomial time, without imposing any restrictions on the underlying graph or initial colouring.

The key tool we use to prove these tractability results is a theorem which allows us to consider only spanning trees of the graph $G$ in order to determine the minimum number of moves required to flood it.  Clearly this does not immediately allow us to solve \textsc{Free-Flood-It}, as the problem remains hard even on trees, and a graph will in general have an exponential number of spanning trees.  However, the result does provide a very useful method for reasoning about the behaviour of sequences of flooding operations on arbitrary graphs.

We begin in Section \ref{notation} with some notation and definitions, then in Section \ref{trees} we prove our result about spanning trees, and a number of corollaries.  Section \ref{poly} is concerned with graphs containing only a polynomial number of connected subgraphs, and in Section \ref{linking} we consider the complexity of $k$-\textsc{Linking-Flood-It}.

\section{Notation and Definitions}
\label{notation}

Suppose the game is played on a graph $G=(V,E)$, with an initial colouring $\omega$ (not necessarily proper) using colours from the \emph{colour-set} $C$.  Each move $m=(v,d)$ then involves choosing some vertex $v \in V$ and a colour $d \in C$, and assigning colour $d$ to all vertices in the same monochromatic component as $v$.  The goal is to make every vertex in $G$ the same colour, using as few moves as possible. 
We define $m_G(G,\omega,d)$ to be the minimum number of moves required in the free variant to give all vertices of $G$ colour $d$, and $m_G(G,\omega)$ to be $\min_{d \in C}m(G,\omega,d)$.

Let $A$ be any subset of $V$.  We set $m_G(A,\omega,d)$ to be the minimum number of moves we must play in $G$ (with initial colouring $\omega$) to create a monochromatic component of colour $d$ that contains every vertex in $A$, and $m_G(A,\omega)=\min_{d \in C}m_G(A,\omega,d)$.  We write $\omega|_A$ for the colouring $\omega$ restricted to the subset $A$.  We say a move $m = (v,d)$ is \emph{played in} $A$ if $v \in A$, and that $A$ is \emph{linked} if it is contained in a single monochromatic component.  Subsets $A,B \subseteq V$ are \emph{adjacent} if there exists $ab \in E$ with $a \in A$ and $b \in B$.  We will use the same notation when referring to (the vertex-set of) a subgraph $H$ of $G$ as for a subset $A \subseteq V(G)$.

For any vertex $v \in V$, we write $\comp_G(v,\omega)$ to denote the monochromatic component of $G$, with respect to $\omega$, that contains $v$.  Given any sequence of moves $S$ on a graph $G$ with initial colouring $\omega$, we denote by $S(\omega,G)$ the new colouring obtained by playing $S$ in $G$.

\section{Spanning trees}
\label{trees}

In this section we investigate the relationship between the number of moves required to flood a connected graph $G$ and the number of moves required to flood spanning trees of $G$.  For any connected graph $G$, let $\mathcal{T}(G)$ denote the set of all spanning trees of $G$.  We prove the following result.

\begin{thm}
Let $G$ be a connected graph with colouring $\omega$ from colour-set $C$.  Then, for any $d \in C$,
$$m_G(G,\omega,d) = \min_{T \in \mathcal{T}(G)} m_T(T,\omega,d).$$
\label{min-tree}
\end{thm}

Since it remains hard to solve 3-\textsc{Free-Flood-It} on trees, this result does not imply that the number of moves required to flood a graph with only a polynomial number of spanning trees can be computed in polynomial time.  However, this equality gives rise to a number of corollaries, proved later in this section, which are then applied to give polynomial-time algorithms for various flood-filling problems in Sections \ref{poly} and \ref{linking}.

To prove the theorem, we demonstrate that the inequalities $m_G(G,\omega,d) \geq \min_{T \in \mathcal{T}(G)} m_T(T,\omega,d)$ and $m_G(G,\omega,d) \leq \min_{T \in \mathcal{T}(G)} m_T(T,\omega,d)$ must both hold.  We begin with an auxiliary result which is crucial to the proof of both inequalities, as it allows us to consider independently optimal sequences to flood disjoint subtrees of a tree.

\begin{lma}
Let $T$ be a tree, with colouring $\omega$ from colour-set $C$, and let $A$ and $B$ be disjoint subsets of $V(T)$ such that $V(T) = A \cup B$ and $T[A], T[B]$ are connected.  Then, for any $d \in C$,
$$m_T(T,\omega,d) \leq m_{T[A]}(A,\omega|_A,d) + m_{T[B]}(B,\omega|_B,d).$$
\label{tree-splitting}
\end{lma}
\begin{proof}
First observe that the result is trivially true if either $A$ or $B$ is empty, so assume $A,B \neq \emptyset$.  We say that a sequence of moves $S$ played in a graph $G$ with colouring $\omega$ \emph{affects} the subset $X$ of $V(G)$ if playing $S$ in $G$ (with colouring $\omega$) changes the colour of at least one vertex in $X$ at some stage.  If $X$ and $Y$ are disjoint subsets of $G$, and $S_X$ and $S_Y$ are sequences of moves played in $X$ and $Y$ respectively, we say there is a \emph{compatible ordering} of $(X,S_X)$ and $(Y,S_Y)$ if at least one of these sequences does not affect the other subset (i.e.~either playing $S_X$ in $G$ does not change the colour of any vertex of $Y$, or else playing $S_Y$ in $G$ does not change the colour of any vertex in $X$).

In fact we prove the following, stronger claim.
\begin{claim}
Let $A$, $B$ and $T$ be as above, and suppose the sequence $S_A$ floods $A$ with colour $d$, and that $S_B$ floods $B$ with some colour $d' \in C$.  Then, if there is no compatible ordering of $(A,S_A)$ and $(B,S_B)$, we have 
$$m_T(T,\omega,d') \leq |S_A| + |S_B|.$$
\end{claim}

To see that the claim implies the result, we set $d = d'$, and choose $S_A$ and $S_B$ to be optimal sequences to flood $T[A]$ and $T[B]$ respectively with colour $d$.  If there is no compatible ordering of $(A,S_A)$ and $(B,S_B)$, the result then follows immediately from the claim (as $|S_A| + |S_B| = m_{T[A]}(A,\omega,d) + m_{T[B]}(B,\omega,d)$), so it suffices to consider the case in which there does exist a compatible ordering of $(A,S_A)$ and $(B,S_B)$.  Without loss of generality, we may assume that playing $S_A$ in $T$ does not change the colour of any vertex in $B$.

Note that if $T'$ is a subtree of the tree $T$, and $S$ is a sequence of moves played in $T'$, then $S$ has exactly the same effect on $T'$ whether the moves are played in the isolated subtree $T'$ or in the larger tree $T$, as no vertices in $T'$ can be linked in $T$ before they are linked in $T'$.  Thus playing $S_A$ in $T$ links $A$ in colour $d$ and (by assumption) does not change the colour of any vertex in $B$, so then playing $S_B$ in the resulting coloured tree gives $B$ colour $d$.  If this also changes the colour of $A$, some vertex $b \in B$ is linked to $A$ and so $A$ must thereafter have the same colour as $b$, which ends up with colour $d$.  Thus the sequence $S_AS_B$ gives all of $T$ colour $d$, implying
$$m_T(T,\omega,d) \leq |S_A| + |S_B| = m_{T[A]}(A,\omega|_A,d) + m_{T[B]}(B,\omega|_B,d),$$
as required.

We now prove the claim by induction on $|B|$.  Note that we may assume without loss of generality that $\omega|_B$ is a proper colouring of $B$, otherwise we may contract monochromatic components.  Suppose $|B| = 1$.  If there is no compatible ordering, then $S_A$ must change the colour of the only vertex in $B$ (linking it to some $a \in A$), and so playing $S_A$ in $T$ makes the whole tree monochromatic with colour $d$.  Thus $m_T(T,\omega,d) \leq |S_A|$, and
$$m_T(T,\omega,d') \leq m_T(T,\omega,d) + 1 \leq |S_A| + 1.$$
But the fact there is no compatible ordering also implies that $|S_B| \geq 1$, otherwise playing the empty sequence $S_B$ in $T$ would not affect $A$.  Hence
$$m_T(T,\omega,d') \leq |S_A| + 1 \leq |S_A| + |S_B|,$$
as required.

Now suppose $|B| > 1$, so $B$ is not monochromatic initially, and assume that the claim holds for smaller $B$.  By the reasoning above, this implies that the theorem holds in these cases: in other words, whenever $\bar{T}$ is a tree and $\bar{A}$ and $\bar{B}$ are disjoint subsets of $V(\bar{T})$ such that $V(\bar{T}) = \bar{A} \cup \bar{B}$, $\bar{T}[\bar{A}]$ and $\bar{T}[\bar{B}]$ are connected, and $|\bar{B}| < |B|$, for any $\bar{d} \in C$ we have 
\begin{equation}
m_{\bar{T}}(\bar{T},\omega|_{\bar{T}}, \bar{d}) \leq m_{\bar{T}[\bar{A}]}(\bar{A}, \omega|_{\bar{A}}, \bar{d}) + m_{\bar{T}[\bar{B}]}(\bar{B},\omega|_{\bar{B}},\bar{d}).
\label{ind-eqn}
\end{equation}

Set ${S_B}'$ to be the initial segment of $S_B$, up to and including the move that first makes $B$ monochromatic (in any colour), so any final moves that simply change the colour of $B$ are omitted.  We may, of course, have ${S_B}' = S_B$, if $B$ is not monochromatic before the final move of $S_B$.

By assumption, there is no compatible ordering of $(A,S_A)$ and $(B,S_B)$.  First suppose that there is no compatible ordering of $(A,S_A)$ and $(B,{S_B}')$ either, and that ${S_B}'$ gives $B$ colour $d'' \in C$ (note that this must be the situation if $S_B = {S_B}'$, in which case $d'' = d'$).  Before the final move of ${S_B}'$ there are $r \geq 2$ monochromatic components in $B$ (all but one of which have colour $d''$), with vertex-sets $B_1, \ldots, B_r$.  For $1 \leq i \leq r$, set $S_i$ to be the subsequence of ${S_B}'$ consisting of moves played in $B_i$, and note that these subsequences partition ${S_B}'$.  Observe also that playing $S_i$ in $T[B_i]$ gives $B_i$ colour $d''$, so $m_{T[B_i]}(B_i,\omega|_{B_i},d'') \leq |S_i|$.

Let $B_1$ be the unique component adjacent to $A$, and set $T_1 =  T[A \cup B_1]$.  Note that $S_A$ floods $T[A]$ with colour $d$, and $S_1$ floods $T[B_1]$ with colour $d''$.  Moreover, as there is no compatible ordering of $(A,S_A)$ and $(B,{S_B}')$, there cannot be a compatible ordering of $(A,S_A)$ and $(B_1,S_1)$: any move from ${S_B}'$ that changes the colour of a vertex in $A$ must belong to $S_1$, and if $S_A$ changes the colour of any vertex in $B$ it must change the colour of at least one vertex in $B_1$.  Thus we can apply the inductive hypothesis to see that 
$$m_{T_1}(T_1, \omega|_{T_1}, d'') \leq |S_A| + |S_1|.$$
Now suppose without loss of generality that $B_2$ is adjacent to $B_1$.  We can then apply (\ref{ind-eqn}) to $T_2 = T[V(T_1) \cup B_2]$ to see that
$$m_{T_2}(T_2, \omega|_{T_2}, d'') \leq m_{T_1}(T_1,\omega|_{T_1},d'') + m_{T[B_2]}(B_2,\omega|_{B_2},d'') \leq |S_A| + |S_1| + |S_2|.$$
Continuing in this way, each time adding an adjacent component, we see that
$$m_T(T,\omega,d'') \leq |S_A| + \sum_{i=1}^r |S_i| = |S_A| + |{S_B}'|.$$
Now, if ${S_B}' = S_B$, this immediately gives the desired result, as $d'' = d'$.  Otherwise, note that $|S_B| \geq |{S_B}'|+1$ and so
$$m_T(T,\omega,d') \leq m_T(T,\omega,d'') + 1 \leq |S_A| + |{S_B}'| + 1 \leq |S_A| + |S_B|,$$
as required.

It remains to consider the case in which there exists a compatible ordering of $(A,S_A)$ and $(B,{S_B}')$.  By assumption, there is no compatible ordering of $(A,S_A)$ and $(B,S_B)$, so playing $S_A$ in $T$ must change the colour of some vertex in $B$. Thus it must be that playing ${S_B}'$ in $T$ does not change the colour of any vertex in $A$, although playing $S_B$ does.  Let $\alpha$ be the first move of $S_B$ which, when the sequence is played in $T$ with initial colouring $\omega$, changes the colour of some vertex in $A$.  Set $\tilde{S_B}$ to be the initial segment of $S_B$ up to but not including $\alpha$ (so ${S_B}'$ and $\tilde{S_B}$ are both proper initial segments of $S_B$, with $\tilde{S_B}$ possibly longer than ${S_B}'$).  Note that $\tilde{S_B}$ makes $B$ monochromatic, and observe also that, as $\alpha$ is played in $B$ but changes the colour of some vertex $a \in A$, $\tilde{S_B}$ must link $a$ and $B$.  

Now suppose we play $\tilde{S_B}$ in $T$, followed by $S_A$.  By definition, $\tilde{S_B}$ does not change the colour of any vertex in $A$, so playing $S_A$ in $T$ with colouring $\tilde{S_B}(\omega,T)$ will make $A$ monochromatic with colour $d$.  However, after playing $\tilde{S_B}$, every vertex in $B$ belongs to the same monochromatic component as $a \in A$, and so the sequence $\tilde{S_B}S_A$ must also give every vertex in $B$ colour $d$.  Thus we see that $m_T(T,\omega,d) \leq |\tilde{S_B}| + |S_A|$, and so
$$m_T(T,\omega,d') \leq m_T(T,\omega,d) + 1 \leq |S_A| + |\tilde{S_B}| + 1 \leq |S_A| + |S_B|,$$
completing the proof of the claim.
\end{proof}

This result is easily applied to give a lower bound on $m_G(G,\omega,d)$.

\begin{lma}
Let $G$ be a connected graph with colouring $\omega$ from colour-set $C$.  Then, for any $d \in C$, 
$$\min_{T \in \mathcal{T}(G)} m_T(T,\omega,d) \leq m_G(G,\omega,d).$$
\label{T<=G}
\end{lma}
\begin{proof}
It suffices to show that there exists a spanning tree $T$ for $G$ such that $m_T(T,\omega,d) \leq m_G(G,\omega,d)$.  We prove this by induction on $m_G(G,\omega,d)$.  

Note that the base case, for $m_G(G,\omega,d) = 0$, is trivial: any spanning tree will do, as all spanning trees are already monochromatic with colour $d$.  So assume $m_G(G,\omega,d) \geq 1$ and let $S$ be an optimal sequence to flood $G$ with colour $d$.

First suppose that $G$ is monochromatic in some colour $d'$ before the final move of $S$, which simply changes the graph's colour to $d$, so $|S| = 1 + m_G(G,\omega,d')$.  Then, by the inductive hypothesis, there exists a spanning tree $T$ for $G$ such that $m_T(T,\omega,d') \leq m_G(G,\omega,d')$.  But we then see that
$$m_T(T,\omega,d) \leq 1 + m_T(T,\omega,d') \leq 1 + m_G(G,\omega,d') = |S| = m_G(G,\omega,d),$$
and so $T$ is the spanning tree we require.

Thus we may assume that, immediately before the final move of $S$, $G$ has $r \geq 2$ monochromatic components with vertex sets $A_1, \ldots, A_r$ (all but one of which have colour $d$), and that the final move is played in $A_1$.  For $1 \leq i \leq r$, let $S_i$ be the subsequence of $S$ consisting of moves played in $A_i$; note that these sequences partition $S$, and that $m_{G[A_i]}(A_i,\omega|_{A_i},d) \leq |S_i|$ for each $i$ (as $S_i$, played in $G[A_i]$, makes this subgraph monochromatic with colour $d$).

For $2 \leq i \leq r$, we know that $|S_i| < |S|$ (as at least one move is played in $A_1$), so we may apply the inductive hypothesis immediately to see that there exists a spanning tree $T_i$ for $G[A_i]$ such that $m_{T_i}(T_i,\omega|_{A_i},d) \leq m_{G[A_i]}(A_i,\omega|_{A_i},d) \leq |S_i|$.  Now observe that, if $S_1'$ is the subsequence $S_1$ with just the final move omitted, $S_1'$ makes $A_1$ monochromatic with some colour $d'$, and $|S_1'| < |S|$.  Thus we can apply the inductive hypothesis here to see that there exists a spanning tree $T_1$ for $G[A_1]$ such that $m_{T_1}(T_1,\omega|_{A_1},d') \leq m_{G[A_1]}(A_1,\omega|_{A_1},d') \leq |S_1'|$ and so $m_{T_1}(T_1,\omega|_{A_1},d) \leq 1 + m_{T_1}(T_1,\omega|_{A_1},d') \leq 1 + |S_1'| = |S_1|$.

Note that, for $2 \leq i \leq r$, $A_1$ must be adjacent to $A_i$, and so there exists an edge $e_i$ in $G$ between $A_1$ and $A_i$.  Define $T$ to be $\bigcup_{i=1}^r T_i \cup \{e_i: 2 \leq i \leq r\}$, and observe that $T$ is a spanning tree for $G$.  Moreover, by Lemma \ref{tree-splitting} we see that
\begin{align*}
m_T(T,\omega,d) & \leq \sum_{i=1}^r m_{T_i}(T_i,\omega|_{A_i},d) \\
			    & \leq \sum_{i=1}^r |S_i| \\
			    & = |S| \\
			    & = m_G(G,\omega,d),
\end{align*}
so $T$ is as required.
\end{proof}

We now proceed to show the reverse inequality.  We call a spanning tree $T$ of $G$ \emph{$d$-minimal} (with respect to the colouring $\omega$) if $m_{T}(T,\omega,d) = \min_{T' \in \mathcal{T}(G)} m_{T'}(T',\omega,d)$, and say that a spanning tree $T$ \emph{preserves monochromatic components of $G$} (with respect to $\omega$) if $T$ and $G$ have the same monochromatic components, i.e.~$\comp_G(v,\omega) = \comp_T(v,\omega)$ for all $v \in V(G)$.  

We shall demonstrate that, for any $d \in C$, there exists a $d$-minimal spanning tree $T$ that preserves the monochromatic components of $G$, and that for such a tree we must have $m_T(T,\omega,d) \geq m_G(G,\omega,d)$.  Our first step is to show that, given any tree $T$ and an edge $e \notin E(T)$, we can replace $T$ with another tree $T'$ that contains $e$, without increasing the number of moves we need to flood the tree.

\begin{lma}
Let $T$ be a tree with colouring $\omega$ from colour-set $C$, and suppose $e=uv \notin E(T)$.  Then, for any $d \in C$, there exists a spanning tree $T'$ of $T \cup \{e\}$, with $e \in E(T')$, such that $m_{T'}(T',\omega,d) \leq m_T(T,\omega,d)$.
\label{add-e}
\end{lma}
\begin{proof}
We proceed by induction on $m_T(T,\omega,d)$.  If $m_T(T,\omega,d) = 0$ then $T$ is already monochromatic with colour $d$, and we can choose $T'$ to be any spanning tree of $T \cup \{e\}$ with $e \in E(T')$.  So assume $m_T(T,\omega,d) > 0$, and let $S$ be an optimal sequence to flood $T$ with colour $d$.

First suppose that $T$ is monochromatic in some colour $d'$ before the final move of $S$, and so this last move just changes the colour to $d$.  We can apply the inductive hypothesis to see that there exists a spanning tree $T'$ of $T \cup \{e\}$, with $e \in E(T')$, such that $m_{T'}(T',\omega,d') \leq m_T(T,\omega,d')$.  But then we have
$$m_{T'}(T',\omega,d) \leq 1 + m_{T'}(T',\omega,d') \leq 1 + m_T(T,\omega,d') = m_T(T,\omega,d).$$

Thus we may assume that the last move links $r \geq 2$ monochromatic subtrees of $T$ with vertex sets $A_1, \ldots, A_r$, and is played in $A_1$.  Note that for each $i$ the subsequence $S_i$ of $S$, consisting of moves played in $A_i$, floods the subtree $T[A_i]$ with colour $d$, and so $m_{T[A_i]}(A_i,\omega|_{A_i},d) \leq |S_i|$.  Observe also that $S_1, \ldots, S_r$ partition $S$.

Suppose first that $u,v \in A_i$ for some $i$.  As at least one move is played in $T_1$, we have $m_{T[A_j]}(A_j,\omega|_{T_j},d) < |S| = m_T(T,\omega,d)$ for $2 \leq j \leq r$, so if $i \neq 1$ we may apply the inductive hypothesis to see that there exists a spanning tree $T_i$ for $T[A_i] \cup \{e\}$, with $e \in E(T_i)$, such that $m_{T_i}(T_i, \omega|_{A_i},d) \leq m_{T[A_i]}(A_i, \omega|_{A_i},d) \leq |S_i|$.

For the case $i=1$, observe that $A_1$ is monochromatic in some colour $d' \neq d$ before the final move of $S_1$, so $m_{T[A_1]}(A_1,\omega|_{A_1},d') < |S_1| \leq |S| = m_T(T,\omega,d)$.  Thus, by the inductive hypothesis, there again exists a spanning tree $T_1$ for $T[A_1] \cup \{e\}$, with $e \in E(T_1)$, such that $m_{T_1}(T_1,\omega|_{A_1},d') \leq m_{T[A_1]}(A_1,\omega|_{A_1},d')$, implying 
$$m_{T_1}(T_1,\omega|_{T_1},d) \leq 1 + m_{T_1}(T_1,\omega|_{T_1},d') \leq 1 + m_{T[A_1]}(A_1,\omega|_{A_1},d') \leq |S_1|.$$

Now let us define $T'$ to be the tree with vertex-set $V(T)$ and edge-set $(E(T) \setminus E(T[A_i])) \cup E(T_i)$.  $T'$ clearly contains $e$, and by Lemma \ref{tree-splitting} we see that
\begin{align*}
m_{T'}(T',\omega,d) & \leq  m_{T_i}(T_i,\omega|_{T_i},d) + \sum_{j \neq i} m_{T[A_j]}(A_j,\omega|_{A_j},d) \\
                    & \leq \sum_{i = 1}^r |S_i| \\
                    & = |S| \\
                    & = m_T(T,\omega,d),
\end{align*}
so $T'$ has the required properties.

It remains to consider the case that the endpoints of $e$ do not lie in the same subtree.  In particular, at most one of $u$ and $v$ lies in $A_1$, so suppose without loss of generality that $v \in A_2$.  Suppose $f$ is the unique edge in $T$ that joins $A_1$ to $A_2$, and set $T'$ to be the tree with vertex-set $V(T)$ and edge-set $(E(T) \setminus \{f\}) \cup \{e\}$.  Once again, it is clear that $T'$ is a tree containing $e$, and we can apply Lemma \ref{tree-splitting} to see that 
\begin{align*}
m_{T'}(T',\omega,d) & \leq \sum_{i=1}^r m_{T[A_i]}(A_i,\omega|_{A_i},d) \\
                    & \leq \sum_{i=1}^r |S_i| \\
                    & = |S| \\
                    & = m_T(T,\omega,d),
\end{align*}
completing the proof.
\end{proof}

Next we show that every coloured graph has a $d$-minimal spanning tree that preserves monochromatic components. 

\begin{lma}
Let $G=(V,E)$ be a connected graph with colouring $\omega$ from colour-set $C$.  Then, for any $d \in C$, there exists a $d$-minimal spanning tree for $G$ that preserves monochromatic components of $G$ with respect to $\omega$.
\label{tree-preserves-cmpts}
\end{lma}
\begin{proof}
We proceed by induction on $|E|$.  If $|E| = |V| - 1$ then $G$ is a tree and the result is trivially true, so suppose $|E| \geq |V|$.    If $\omega$ is a proper colouring of $G$ then any spanning tree preserves its monochromatic components, so we may assume that there exists an edge $e \in E$ such that both endpoints of $e$ receive the same colour under $\omega$.

By Lemma \ref{add-e}, there exists a $d$-minimal spanning tree $T_0$ of $G$ that contains $e$.  Let $T_1$ and $G_1$ be the graphs obtained from $T_0$ and $G$ respectively by contracting the edge $e$, and let $\omega_1$ be the corresponding colouring of $V(T_1) = V(G_1)$.  Note that $T_1$ is a spanning tree for $G_1$, and we have $m_{G_1}(G_1, \omega_1, d) = m_G(G,\omega,d)$ and $m_{T_1}(T_1,\omega_1,d) = m_{T_0}(T_0,\omega,d)$.

Since $e(G_1) < e(G)$, we may apply the inductive hypothesis to $G_1$ to find a $d$-minimal spanning tree $T_2$ of $G_1$ that preserves monochromatic components of $G_1$ with respect to $\omega_1$.  Let $T$ be a spanning tree of $G$ obtained from $T_2$ by uncontracting $e$ (note that this tree is not necessarily unique).

It follows immediately that $T$ preserves monochromatic components of $G$ with respect to $\omega$, so it remains to check that $T$ is $d$-minimal.  It is clear that $m_T(T,\omega,d) = m_{T_2}(T_2,\omega_1,d)$, so we see that
\begin{align*}
m_T(T,\omega,d) & = m_{T_2}(T_2,\omega_1,d) \\
				& = \min_{T' \in \mathcal{T}(G_1)}(T',\omega_1,d) \\
                & \leq m_{T_1}(T_1,\omega_1,d) \\
                & = m_{T_0}(T_0,\omega,d) \\
                & = \min_{T' \in \mathcal{T}(G)}(T',\omega,d),
\end{align*}
as required.                
\end{proof}

We are now ready to prove our upper bound on $m_G(G,\omega,d)$.

\begin{lma}
Let $G$ be a connected graph with colouring $\omega$ from colour-set $C$.  Then, for any $d \in C$,
$$m_G(G,\omega,d) \leq \min_{T \in \mathcal{T}(G)} m_T(T,\omega,d).$$
\label{G<=T}
\end{lma}
\begin{proof}
It suffices to prove that, for any $T \in \mathcal{T}(G)$, $m_G(G,\omega,d) \leq m_T(T,\omega,d)$.  We proceed by induction on $m_T(T,\omega,d)$.  If $m_T(T,\omega,d) = 0$ the result is trivially true, as $G$ must already be monochromatic with colour $d$, so assume $m_T(T,\omega,d)>0$. By Lemma \ref{tree-preserves-cmpts}, there exists a $d$-minimal spanning tree $T'$ of $G$ that preserves the monochromatic components of $G$ with respect to $\omega$.  Note that $m_T(T,\omega,d) \geq m_{T'}(T',\omega,d)$.

Let $S$ be an optimal sequence to flood the tree $T'$ (considered in isolation), and let $\alpha = (v,d')$ be the first move of $S$.  Note that $m_{T'}(T',\alpha(\omega,T'),d) < m_T(T,\omega,d)$.  By choice of $T'$, we have $\comp_G(v,\omega) = \comp_{T'}(v,\omega)$ for every $v \in V(T')$, and so exactly the same vertices have their colour changed to $d'$ when $\alpha$ is played in $G$ as when it is played in $T'$.  Thus $\alpha(\omega,G) = \alpha(\omega,T')$.  Note that $T'$ does not necessarily preserve monochromatic components of $G$ with respect to the new colouring $\alpha(\omega,T')$, but we are nevertheless able to apply the inductive hypothesis to the graph with this colouring.

Thus we see that
\begin{align*}
m_G(G,\omega,d) & \leq 1 + m_G(G,\alpha(\omega,G),d) \\
                & = 1 + m_G(G,\alpha(\omega,T'),d) \\
                & \leq 1 + m_{T'}(T',\alpha(\omega,T'),d) & \mbox{by inductive hypothesis} \\
                & = m_{T'}(T',\omega,d) \\
                & \leq m_T(T,\omega,d),
\end{align*}
as required.                
\end{proof}

Together with the preceding results, this proves our main theorem.

\begin{proof}[Proof of Theorem \ref{min-tree}]
The proof follows immediately from Lemma \ref{T<=G} and Lemma \ref{G<=T}.
\end{proof}

We now prove five corollaries of Theorem \ref{min-tree}, which will be useful in the following sections.  In the first two of these, we exploit Theorem \ref{min-tree} to generalise Lemma \ref{tree-splitting} very substantially.  We begin by showing that the result can be extended from trees to arbitrary graphs.

\begin{cor}
Let $G$ be a connected graph, with colouring $\omega$ from colour-set $C$, and let $A$ and $B$ be disjoint subsets of $V(G)$ such that $V(G) = A \cup B$ and $G[A], G[B]$ are connected.  Then, for any $d \in C$,
$$m_G(G,\omega,d) \leq m_{G[A]}(A,\omega|_A,d) + m_{G[B]}(B,\omega|_B,d).$$
\label{decomposition}
\end{cor}
\begin{proof}
First note that, by Theorem \ref{min-tree}, there exist spanning trees $T_A$ and $T_B$ of $G[A]$ and $G[B]$ respectively, such that $m_{T_A}(T_A,\omega|_A,d) = m_{G[A]}(A,\omega|_A,d)$ and $m_{T_B}(T_B,\omega,d) = m_{G[B]}(B,\omega|_B,d)$.  As $G$ is connected, there exists an edge $e$ between $A$ and $B$.  Set $T$ to be the tree with vertex set $V(G)$ and edge set $E(T_A) \cup E(T_B) \cup \{e\}$.  By Lemma \ref{tree-splitting}, $m_T(T,\omega,d) \leq m_{T_A}(T_A,\omega|_A,d) + m_{T_B}(T_B,\omega|_B,d)$; but $T$ is a spanning tree for $G$ so, by Theorem \ref{min-tree}, 
\begin{align*}
m_G(G,\omega,d) & \leq m_T(T,\omega,d) \\
				& \leq m_{T_A}(T_A,\omega|_A,d) + m_{T_B}(T_B,\omega|_B,d) \\
				& = m_{G[A]}(A,\omega|_A,d) + m_{G[B]}(B,\omega|_B,d).
\end{align*}				
\end{proof}

We now generalise this result even further, demonstrating that we do not in fact require $A$ and $B$ to be disjoint.

\begin{cor}
Let $G$ be a connected graph, with colouring $\omega$ from colour-set $C$, and let $A$ and $B$ be subsets of $V(G)$ such that $V(G) = A \cup B$ and $G[A], G[B]$ are connected.  Then, for any $d \in C$,
$$m_G(G,\omega,d) \leq m_{G[A]}(A,\omega|_A,d) + m_{G[B]}(B,\omega|_B,d).$$
\label{general-decomposition}
\end{cor}
\begin{proof}
If $A \cap B = \emptyset$ then the result is immediate from Corollary \ref{decomposition}, so assume that $A \cap B = X \neq \emptyset$.  We now construct a new graph $G'$ by ``blowing up'' every vertex in $X$, replacing each $v \in X$ with a pair of adjacent vertices.  We set
$$V(G') = (V(G) \setminus X) \cup \{v_A, v_B: v \in X\},$$ 
and 
\begin{align*}
E(G') = & E(G \setminus X) \\
        & \cup \{v_Av_B: v \in X\} \\
        & \cup \{v_Au, v_Bu: v \in X, u \notin X, uv \in E(G)\} \\
        & \cup \{v_Aw_A,v_Bw_B : vw \in E(G[X])\},
\end{align*}        
and give $G'$ a colouring $\omega'$, where $\omega'(v_A) = \omega'(v_B) = \omega(v)$ for $v \in X$ and $\omega'(u) = \omega(u)$ for $u \notin X$.  Observe that contracting monochromatic components in both $(G,\omega)$ and $(G',\omega')$ will give identical coloured graphs, so we must have 
$$m_G(G,\omega,d) = m_{G'}(G',\omega',d).$$

We then define $A' = A \setminus X \cup \{v_A: v \in X\}$ and $B' = B \setminus X \cup \{v_B: v \in X\}$.  Now, $A'$ and $B'$ partition $V(G')$ and induce connected subgraphs, so we can apply Corollary \ref{decomposition} to see that
$$m_{G'}(G',\omega',d) \leq m_{G[A']}(A',\omega'|_{A'},d) + m_{G[B']}(B',\omega'|_{B'},d).$$
But $G[A']$ with colouring $\omega'|_{A'}$ is identical to $G[A]$ with colouring $\omega|_A$, so $m_{G[A']}(A',\omega'|_{A'},d) = m_{G[A]}(A,\omega|_A,d)$, and similarly $m_{G[B']}(G[B'],\omega'|_{B'},d) = m_{G[B]}(B,\omega|_B,d)$. Thus
$$m_G(G,\omega,d) = m_{G'}(G',\omega',d) \leq m_{G[A]}(A,\omega|_A,d) + m_{G[B]}(B,\omega|_B,d),$$
as required.
\end{proof}

The next two corollaries are concerned with the number of moves required to flood a connected subgraph $H$ of a graph $G$.  We begin by showing that adding additional edges to $H$ cannot increase the number of moves required to flood the graph.

\begin{cor}
Let $G$ be a connected graph with colouring $\omega$ from colour-set $C$, and $H$ a connected spanning subgraph of $G$. Then, for any $d \in C$,
$$m_{G}(G,\omega,d) \leq m_H(H,\omega,d).$$
\label{add-edges}
\end{cor}
\begin{proof}
As $H$ is a connected spanning subgraph of $G$, we have $\mathcal{T}(H) \subseteq \mathcal{T}(G)$.  Thus, by Theorem \ref{min-tree}, 
$$m_{G}(G,\omega,d) = \min_{T \in \mathcal{T}(G)}(T,\omega,d) \leq \min_{T \in \mathcal{T}(H)}(T,\omega,d) = m_H(H,\omega,d).$$
\end{proof}

Now we consider the case in which $H$ is an arbitrary subgraph, and show that the number of moves we must play in $G$ to link the vertices of $H$ is at most the number required to flood the isolated subgraph $H$.

\begin{cor}
Let $G$ be a connected graph with colouring $\omega$ from colour-set $C$, and $H$ a connected subgraph of $G$.  Then, for any $d \in C$,
$$m_G(V(H),\omega,d) \leq m_H(H,\omega|_H,d).$$
\label{subgraph}
\end{cor}
\begin{proof}
Set $H'$ to be the subgraph of $G$ induced by $\bigcup_{v \in V(H)} \comp_G(v,\omega)$, and note that a sequence of moves played in $G$ floods $H'$ if and only if it floods $H$ (with the same colour).  Observe that we can add edges to $H$ to obtain a coloured graph equivalent to $H'$ (when both graphs have colouring inherited from $\omega$): if an edge is added in $H$ between every pair of vertices that belong to either the same monochromatic component or adjacent monochromatic components in $G$, then contracting monochromatic components in this new graph will give the same result as contracting monochromatic components of $H'$.  Thus, by Corollary \ref{add-edges}, $m_{H'}(H',\omega|_{H'},d) \leq m_H(H,\omega|_H,d)$.

We proceed to prove the inequality by induction on $m_{H'}(H',\omega,d)$.  Note that if $m_{H'}(H',\omega,d)=0$ then the result is trivially true, so assume that $m_{H'}(H',\omega,d)>0$.  Let $S$ be an optimal sequence to flood the (isolated) subgraph $H'$ with colour $d$, and suppose the first move of $S$ is $\alpha$.  As we have $\comp_{H'}(v,\omega) = \comp_G(v,\omega)$ for every $v \in V(H')$, the move $\alpha$ changes the colour of exactly the same vertices when played in $H'$ as it does when played in the larger graph $G$.  Thus $\alpha(\omega,G)|_{H'} = \alpha(\omega|_{H'},H')$.  

Note that $m_{H'}(H',\omega|_{H'},d) = 1 + m_{H'}(H',\alpha(\omega|_{H'},H'),d)$, so 
$$m_{H'}(H',\alpha(\omega,G)|_{H'},d) = m_{H'}(H',\alpha(\omega|_{H'},H'),d) < m_H(H,\omega|_H,d),$$  
and we can apply the inductive hypothesis to see that
\begin{equation}
m_G(V(H'),\alpha(\omega,G),d) \leq m_{H'}(H',\alpha(\omega,G)|_{H'},d).
\label{subg-ind}
\end{equation}
Hence
\begin{align*}
m_G(V(H),\omega,d) & = m_G(V(H'),\omega,d) \\
                & \leq 1 + m_G(V(H'),\alpha(\omega,G),d) \\
                & \leq 1 + m_{H'}(H',\alpha(\omega,G)|_{H'},d) \\ 
                & \qquad \qquad \qquad \qquad \qquad \qquad \mbox{by (\ref{subg-ind})}\\
                & = 1 + m_{H'}(H',\alpha(\omega|_{H'},H'),d) \\
                & \qquad \qquad \qquad \qquad \qquad \qquad \mbox{as $\alpha(\omega,G)|_{H'} = \alpha(\omega|_{H'},H')$} \\
                & = m_{H'}(H',\omega|_{H'},d) \\
                & \leq m_H(H,\omega|_H,d),
\end{align*} 
as required.
\end{proof}

Finally, we consider the number of moves required to connect a given subset of the vertices of $G$.  For any $U \subseteq V(G)$, let $\mathcal{T}(U,G)$ be the set of all subtrees $T$ of $G$ such that $U \subseteq V(T)$.  We then characterise the number of moves required to link $U$ in terms of the number of moves required to flood elements of $\mathcal{T}(U,G)$.

\begin{cor}
Let $G$ be a connected graph with colouring $\omega$ from colour-set $C$, and let $U \subseteq V(G)$.  Then, for any $d \in C$, 
$$m_G(U,\omega,d) = \min_{T \in \mathcal{T}(U,G)} m_T(T,\omega|_T,d).$$
\label{set-tree}
\end{cor}
\begin{proof}
We begin by showing that $m_G(U,\omega,d) \leq \min_{T \in \mathcal{T}(U,G)} m_T(T,\omega|_T,d)$.  Let $T \in \mathcal{T}(U,G)$.  Then, as $U \subseteq V(T)$, we see by definition of $m_G(U,\omega,d)$ that $m_G(U,\omega,d) \leq m_G(V(T),\omega,d)$.  Moreover, by Corollary \ref{subgraph} we know that $m_G(V(T),\omega,d) \leq m_T(T,\omega|_T,d)$, and so $m_G(U,\omega,d) \leq m_T(T,\omega|_T,d)$.  As this holds for any $T \in \mathcal{T}(U,G)$, it follows that 
$$m_G(U,\omega,d) \leq \min_{T \in \mathcal{T}(U,G)} m_T(T,\omega|_T,d).$$

To show the reverse inequality, suppose that $S$ is an optimal sequence to link $U$ in $G$.  Let $A$ be the vertex set of the monochromatic component, with respect to $S(\omega,G)$, that contains $U$, and let $S'$ be the subsequence of $S$ consisting of moves played in $A$.  Then $S'$ makes $G[A]$ monochromatic with colour $d$, and so we have
$$m_G(U,\omega,d) = |S| \geq |S'| \geq m_{G[A]}(A,\omega|_A,d).$$
By Theorem \ref{min-tree}, there exists a spanning tree $T_A$ for $A$ with $m_A(A,\omega|_A,d) = m_{T_A}(T_A,\omega|_{T_A},d)$.  But then $U \subseteq V(T)$, so $T_A \in \mathcal{T}(U,G)$ and we have 
$$m_G(U,\omega,d) \geq m_{G[A]}(A,\omega|_A,d) = m_{T_A}(T_A,\omega|_{T_A},d) \geq \min_{T \in \mathcal{T}(U,G)} m_{T}(T,\omega|_{T},d).$$
\end{proof}

In summary, we see from Corollary \ref{general-decomposition} that the number of moves required to flood a graph is bounded above by the sum of the numbers of moves required to flood connected subgraphs which cover the vertex set, whereas Corollaries \ref{add-edges} and \ref{subgraph} show that adding edges and vertices to a connected graph $H$ can only decrease the number of moves required to make the vertex-set of $H$ monochromatic.  Corollary \ref{set-tree} allows us to calculate the minimum number of moves required to connect some subset of the vertices by considering the number of moves required to flood subtrees of $G$.

\section{Graphs with polynomial bounds on the numbers of connected subgraphs}
\label{poly}

Given a vertex $v$ in an arbitrary graph $G$, the number of possible values of $\comp_G(v,\omega)$, as $\omega$ ranges over all possible colourings of $G$, will in general be exponential.  However, it is clear that $\comp_G(v,\omega)$ must be a connected subgraph of $G$ containing $v$, and in some interesting classes of graphs the number of connected subgraphs containing any given vertex is bounded by a polynomial function of $|G|$.  In this section we discuss polynomial time algorithms to solve flood-filling problems in this situation.

First, in Section \ref{free}, we apply corollaries of Theorem \ref{min-tree} to show that \textsc{Free-Flood-It} can be solved in polynomial time on graphs which have only a polynomial number of connected subgraphs.  Then, in Section \ref{fixed}, we give a direct proof that the same is true for the fixed variant.

It should be noted, however, that this condition is not \emph{necessary} for a graph to admit a polynomial-time algorithm to solve \textsc{Free-Flood-It}.  $K_n$ has $\Theta(2^n)$ connected induced subgraphs, but the number of moves required to flood the graph in either version of the game is always one fewer than the number of colours used in the initial colouring.  Graphs corresponding to rectangular $2 \times n$ boards give another such example for the fixed case, as there are $\Omega(2^n)$ connected subgraphs containing any given vertex but \textsc{Fixed-Flood-It} can be solved in linear time in this setting \cite{clifford}.

\subsection{The FREE-FLOOD-IT case}
\label{free}

In this section we prove the following theorem.
\begin{thm}
Let $p$ be a polynomial, and let $\mathcal{G}_p$ be the class of graphs such that, for any $G \in \mathcal{G}_p$, the number of connected subgraphs of $G$ is at most $p(|G|)$.  Suppose $G \in \mathcal{G}_p$ has colouring $\omega$ from colour-set $C$.  Then, for any $d \in C$, we can compute $m_G(G,\omega,d)$ in polynomial time, and hence we can also compute $m_G(G,\omega)$ in polynomial time.
\label{poly-areas}
\end{thm}
It is easy to check that, if $G$ is a subdivision of some fixed graph $H$, the number of connected subgraphs of $G$ is bounded by a polynomial function of $|G|$, and so Theorem \ref{poly-areas} implies a conjecture of Meeks and Scott \cite{general}.
\begin{cor}
\textsc{Free-Flood-It} is solvable in polynomial time on subdivisions of any fixed graph $H$.
\label{subdivisions}
\end{cor}

In the next theorem, we give an explicit bound on the time taken to solve \textsc{Free-Flood-It} in terms of the number of connected subgraphs in the graph we are considering.  Theorem \ref{poly-areas} follows immediately from this result.  The proof relies on Corollary \ref{general-decomposition}, which allows us to consider optimal sequences in distinct components of the graph independently.  

\begin{thm}
Let $G$ be a connected graph with colouring $\omega$ from colour-set $C$, and suppose $G$ has at most $N$ connected subgraphs.  Then we can compute $m_G(G,\omega,d)$ for every $d \in C$, and hence $m_G(G,\omega)$, in time $O(|C|^3 \cdot N^3)$.
\end{thm}
\begin{proof}
Note that we may assume without loss of generality that $\omega$ is a proper colouring of $G$, otherwise we can contract monochromatic components to obtain an equivalent coloured graph.  Let $\mathcal{H}$ be the set of connected subgraphs of $G$.  We compute $m_H(H,\omega|_H,d_1)$ recursively, for each $H \in \mathcal{H}$ and $d_1 \in C$.  For any $H \in \mathcal{H}$ we write $(A,B) \in \splitgraph(H)$ if $A$ and $B$ are connected proper subgraphs of $H$ such that $V(A) \cup V(B) = V(H)$ and $V(A) \cap V(B) = \emptyset$.

We define a function $m^*(H,\omega|_H,d_1)$, and claim that for any $H \in \mathcal{H}$ and $d_1 \in C$, we have $m_H(H,\omega|_H,d_1) = m^*(H,\omega|_H,d_1)$.  We first define
\[m^*(\{v\},\omega|_{\{v\}},d_1) = \begin{cases}
									0	& \mbox{if $\omega(v) = d_1$} \\
									1	& \mbox{otherwise.}
								\end{cases}
\]
and observe that this gives $m_H(H,\omega|_H,d_1) = m^*(H,\omega|_H,d_1)$ whenever $|H| = 1$.  Further values of $m^*$ are defined recursively as follows:
\begin{align}
m^*(H,\omega|_H,d_1) = &  \nonumber \\
                  \min \{ & \min_{(A,B) \in \splitgraph(H)} \{m_A(A,\omega|_A,d_1) + m_B(B,\omega|_B,d_1)\}, \nonumber \\
												        & 1 + \min_{\substack{(A,B) \in \splitgraph(H) \\ d_2 \in C}} \{m_A(A,\omega|_A,d_2) + m_B(B,\omega_B,d_2)\}\},
\label{poly-recursion}
\end{align}									       

The fact that $m_H(H,\omega|_H,d_1) \leq m^*(H,\omega|_H,d_1)$ follows from Corollary \ref{decomposition}.  To see the reverse inequality in the case that $|H| > 1$ (and so by assumption $H$ is not monochromatic under $\omega$), we consider the final move $\alpha$ in an optimal sequence to flood $H$ with colour $d_1$: either $\alpha$ changes the colour of some monochromatic area $X$, linking it to monochromatic areas $Y_1, \ldots, Y_r$ which already have colour $d_1$, or else $H$ is already monochromatic in some colour $d_2$ before the final move, and $\alpha$ simply changes its colour to $d_1$.  In the first case, we set $A = Y_1$ and $B = X \cup Y_2 \cup \ldots \cup Y_r$, and note that the disjoint subsequences of $S$ consisting of moves played in $A$ and $B$ respectively flood the relevant subgraphs with colour $d_1$.  Hence $|S| \geq m_A(A,\omega|_A,d_1) + m_B(B,\omega|_B,d_1)$.  In the case that $H$ is monochromatic before $\alpha$, we observe that $H$ cannot be monochromatic before the penultimate move of $S$ (otherwise $S$ would not be optimal) and apply the reasoning above to the initial segment $S'$ of $S$ in which the final move is omitted, a sequence which floods $H$ with colour $d_2$: there exists $(A,B) \in \splitgraph(H)$ such that $|S'| \geq m_A(A,\omega|_A,d_2) + m_B(B,\omega|_B,d_2)$, and hence $|S| \geq 1 + m_A(A,\omega|_A,d_2) + m_B(B,\omega|_B,d_2)$.  Thus in either case we have $m^*(H,\omega|_H,d_1) \leq m_H(H,\omega|_H,d_1)$.

Observe that every subgraph on the right hand side of (\ref{poly-recursion}) contains strictly fewer vertices than $H$, and so a recursion based on this relationship will terminate.  Thus it remains to show that we can calculate $m^*(H,\omega|_H,d_1)$ for all $H \in \mathcal{H}$ and $d_1 \in C$ in time $O(|C|^3 \cdot N^3)$.  

First we need to construct a list of all connected subgraphs of $G$.  Clearly each vertex in the graph is a connected subgraph of order one, and given all connected subgraphs of order $k$ we can construct all connected subgraphs of order $k+1$ by considering all possible ways of adding a vertex.  Thus, if $N_i$ denotes the number of connected subgraphs of order $i$ in $G$, we can construct the list in time
$$ n + \sum_{i=1}^{n-1} N_i(n-i) \leq n \cdot N = O(N^2).$$

To compute $m^*$, we begin by initialising the table in time $O(|G|)$, then all further values of $m^*$ are then calculated as the minimum over combinations of two other entries.  As our table has $N \cdot |C|$ entries, there are at most $N^2 \cdot |C|^2$ combinations we need to consider, and so we can compute all entries in time at most $O(N^3 \cdot |C|^3)$.  This immediately gives $m_G(G,\omega,d_1)$ for each $d_1 \in C$, and to compute $m_G(G,\omega)$ we simply take the minimum over $|C|$ entries.  Thus we can compute both $m_G(G,\omega,d)$ and $m_G(G,\omega)$ in time $O(N^3 \cdot |C|^3)$.
\end{proof}

\subsection{The FIXED-FLOOD-IT case}
\label{fixed}

In this section we show that the fixed variant, \textsc{Fixed-Flood-It}, can be solved in polynomial time on any coloured graph $(G, \omega)$ in which there are a polynomial number of connected subgraphs.  As a special case, this gives an alternative proof of Lagoutte's result \cite{lagoutte} that \textsc{Fixed-Flood-It} $\in \mathbf{P}$ when restricted to cycles, as a cycle has only a quadratic number of connected subgraphs.

When considering the fixed variant of the game, we use the same notation as before but add a superscript to denote the fixed root vertex at which we play, writing for example $m_G^{(v)}(G,\omega,d)$ for the minimum number of moves that must be played at $v$ to give all vertices colour $d$.

\begin{thm}
Let $p$ be a polynomial, and let $\mathcal{G}_p$ be the class of graphs such that, for any $G \in \mathcal{G}_p$, the number of connected subgraphs of $G$ is at most $p(|G|)$.   Suppose $G \in \mathcal{G}_p$ has colouring $\omega$ from colour-set $C$.  Then, for any $d \in C$ and $v \in V(G)$, we can compute $m_G^{(v)}(G,\omega,d)$ in polynomial time, and hence we can also compute $m_G^{(v)}(G,\omega)$ in polynomial time.
\label{poly-areas-fixed}
\end{thm}

\begin{proof}
We define a set of \emph{states} $\mathcal{S}$, where each $S_i \in \mathcal{S}$ is a pair $(A_i,d_i)$ with $d_i \in C$, and $A_i$ a connected subgraph of $G$ containing $v$; we say we are in state $S_i$ if $A_i$ is the maximal monochromatic component containing $v$ and has colour $d_i$.  We now construct a digraph $D$ with vertex-set $\mathcal{S}$ and edge set $E$, where $(A_i,d_i)(A_j,d_j) \in E$ if and only if $A_i \subseteq A_j$, $A_j \setminus A_i$ is either empty or has colour $d_j$ under $\omega$, and no vertex in $\Gamma(A_j) \setminus A_j$ has colour $d_j$.  Thus there is a directed edge from $S_i$ to $S_j$ if and only if we can reach state $S_j$ from state $S_i$ with a single move.  Note that we can construct $D$ in time $O(|\mathcal{S}|^2)$.

Let us denote by $S_0$ the initial state (so the tuple $S_0$ consists of the maximal monochromatic area containing $v$ at the start, and its initial colour under $\omega$), and by $S_d$ the state $(G,d)$ in which the entire graph is monochromatic with colour $d$.  Then the problem of computing $m^{(v)}(G,\omega,d)$ is exactly that of finding the length of a shortest path from $S_0$ to $S_d$ in $D$, which can be done for all $d \in C$ in time $O(|\mathcal{S}|^2)$ (by Dijkstra's algorithm; see \cite{clr90}).  By assumption, $|\mathcal{S}| \leq p(|G|) \cdot |C|$, and so we can construct $D$ and compute $m^{(v)}(G,\omega,d)$ in time $O(p(|G|)^2 \cdot |C|^2)$.  To calculate $m^{(v)}(G,\omega)$ we simply have to take the minimum over $|C|$ values from this computation, so we can calculate $m^{(v)}(G,\omega)$ in time $O(p(|G|)^2 \cdot |C|^2)$.
\end{proof}

\section{The complexity of $k$-LINKING FLOOD IT}
\label{linking}

In this section we use results from Section \ref{trees} to show that $k$-\textsc{Linking-Flood-It}, the problem of determining the minimum number of moves required to link some given set of $k$ points (when moves can be played at any vertex), is solvable in polynomial time for any fixed $k$.

We begin with some additional notation.  Let $U$ be a subset of $V(G)$.  We will say $(U_1,U_2) \in \partition(U)$ if $U_1$ and $U_2$ are disjoint nonempty subsets of $U$ such that $U = U_1 \cup U_2$.  Recall that $\mathcal{T}(U,G)$ is the set of all subtrees $T$ of $G$ such that $U \subseteq V(T)$.  For $1 \leq i \leq |G|$, set $\mathcal{T}_i(U,G) = \{T \in \mathcal{T}(U,G): |T| \leq i\}$.

Recall from Corollary \ref{set-tree} that, for any $U \subseteq V(G)$,
$$m_G(U,\omega,d) = \min_{T \in \mathcal{T}(U,G)} m_T(T,\omega|_T,d).$$
We use this result to give a dynamic programming algorithm to solve $k$-\textsc{Linking-Flood-It} in polynomial time, for any fixed $k$.  Note that the statement of the theorem assumes that the initial colouring of the graph being considered is proper, but of course if this is not the case we can simply contract monochromatic components to obtain an equivalent properly coloured graph.

\begin{thm}
Let $G = (V,E)$ be a connected graph of order $n$, with proper colouring $\omega$ from colour-set $C$, and let $U \subseteq V$ with $|U| = k$.  Then, for any $d \in C$, we can compute $m_G(U,\omega,d)$ in time $O(n^{k+3} \cdot |E| \cdot |C|^2 \cdot 2^k)$.
\label{k-points}
\end{thm}
\begin{proof}
We demonstrate a dynamic programming algorithm to compute values of a function $f$, taking as arguments a nonempty subset $W \subset V$ of at most $k$ vertices, the initial colouring $\omega$ of the graph, a colour $d_1 \in C$, and an index $i \in \{1, \ldots, n\}$.  We will show that, for any values of these arguments, we have
\begin{equation*}
f(W,\omega,d_1,i) = \begin{cases}
                      \min_{\substack{T \in \mathcal{T}_i(W,G)}} m_T(T,\omega|_T,d_1) & \text{if $\mathcal{T}_i(W,G) \neq \emptyset$} \\
                      \infty											  & \text{otherwise.}
                   \end{cases}
\end{equation*}
Thus, as $\mathcal{T}_n(U,G) \neq \emptyset$, we see by Corollary \ref{set-tree} that
$$m_G(U,\omega,d) = \min_{T \in \mathcal{T}(U,G)} m_T(T,\omega|_T,d) = \min_{T \in \mathcal{T}_n(U,G)} m_T(T,\omega|_T,d) = f(U,\omega,d,n).$$

We initialise our table by setting
\begin{align*}
f(W,\omega,d_1,1) = \begin{cases}
						\infty            & \text{if $|W| \geq 2$}\\					
						1				  & \text{if $W=\{w\}$ and $\omega(w) \neq d_1$} \\						
						0				  & \text{if $W=\{w\}$ and $\omega(w) = d_1$,} \\
					  \end{cases} \\
\end{align*}	
and observe that this gives the desired value of $f(W,\omega,d_1,1)$ for all choices of $W$ and $d_1$.

We define further values of $f$ recursively.  First, for any $W$, $\omega$, $d_1$ and $i$, we set
\begin{align*}
\poss(W,\omega,d_1,i) = \{ & ((W_1 \cup \{x_1\}, \omega, d_1, j_1),(W_2 \cup \{x_2\}, \omega, d_1, j_2)):  \\
						   & (W_1,W_2) \in \partition(W), x_1x_2 \in E, x_1 \notin W_2, x_2 \notin W_1, \\
						   & j_1 + j_2 = i, j_1,j_2 > 0\},
\end{align*}
so there is an element of $\poss(W,\omega,d_1,i)$ corresponding to each way of partitioning $W$ into two non-empty subsets, each way of picking an edge in $G$ and associating one endpoint with each subset, and each pair of positive integers summing to $i$.
We then define, for $i \geq 2$,
\begin{equation*}
f_1(W,\omega,d_1,i) = \begin{cases}
						\min_{(\mathbf{z_1},\mathbf{z_2}) \in \poss(W,\omega,d_1,i)} \{f(\mathbf{z_1}) + f(\mathbf{z_2})\} & \text{if $\poss(W,\omega,d_1,i) \neq \emptyset$} \\
						\infty & \text{otherwise,}
					  \end{cases}
\end{equation*}
and
\begin{equation*}
f_2(W,\omega,d_1,i) = 1 + \min_{d_2 \in C} \{f_1(W,\omega,d_2,i)\}.
\end{equation*}
Finally we set
\begin{equation}
f(W,\omega,d_1,i) = \min \{f_1(W,\omega,d_1,i), f_2(W,\omega,d_1,i), f(W,\omega,d_1,i-1)\}.
\label{k-linking-recursion}
\end{equation}

To show that $f$ has the required properties, we first prove by induction on $i$ that we have $f(W,\omega,d_1,i) \leq \min_{T \in \mathcal{T}_i(W,G)} m_T(T,\omega|_T,d_1)$ for each choice of $W$ and $d_1$, if $\mathcal{T}_i(W,G) \neq \emptyset$.  Later we will also prove the reverse inequality.  We have already seen that equality holds in the base case, for $i=1$, so let us consider the case for $i>1$ and assume that the result holds for smaller values.  

If $|W| = 1$, it is clear that 
$$\min_{T \in \mathcal{T}_i(W,G)} m_T(T,\omega|_T,d_1) = \min_{T \in \mathcal{T}_1(W,G)} m_T(T,\omega|_T,d_1) = f(W,\omega,d_1,1).$$
Thus we have
$$\min_{T \in \mathcal{T}_i(W,G)} m_T(T,\omega|_T,d_1) = f(W,\omega,d_1,1) \geq f(W,\omega,d_1,2) \geq \cdots \geq f(W,\omega,d_1,i),$$
as required.  So we may assume $|W| \geq 2$.

We may assume that there exists at least one subtree of $G$ of order at most $i$ that contains the vertices of $W$.  Fix $T \in \mathcal{T}_i(W,G)$ such that $m_T(T,\omega|_T,d_1) = \min_{T' \in \mathcal{T}_i(W,G)} m_{T'}(T',\omega|_{T'},d_1)$ and $|T|$ is minimal.  As $|T| \geq |W| \geq 2$ and $\omega$ is a proper colouring of $G$, $T$ is not monochromatic under $\omega$.  Let $S$ be an optimal sequence to flood $T$ with colour $d_1$.  We proceed by case analysis on the final move, $\alpha$, of $S$.

If $T$ is not monochromatic immediately before $\alpha$, this final move must link $r \geq 2$ monochromatic components, with vertex-sets $A_1, \ldots, A_r$.  We may assume $\alpha$ is played in $A_1$, and so that each of $A_2, \ldots, A_r$ is adjacent only to $A_1$ in $T$.  By minimality of $T$, we cannot have $W \subseteq V(A_l)$ for any $A_l$, so without loss of generality we may assume that $\emptyset \neq W \cap A_2 \neq W$.  

Let $x_1x_2$ be the unique edge of $T$ with $x_1 \in A_1$ and $x_2 \in A_2$.  For $l \in \{1,2\}$, set $T_l$ to be the component of $T - \{x_1x_2\}$ that contains $x_l$, $W_l=W \cap V(T_l)$, and $j_l = |T_l|$.  Note that $((W_1 \cup \{x_1\}, \omega, d_1, j_1),(W_2 \cup \{x_2\}, \omega, d_1, j_2)) \in \poss(W,\omega,d_1,i)$, and also that for $l \in\{1,2\}$, $T_l \in T_{j_l}(W_l \cup \{x_l\},G)$.  Set $S_l$ to be the subsequence of $S$ consisting of moves played in $T_l$, and observe that $S_1$ and $S_2$ partition $S$. Moreover $S_l$, played in $T_l$, makes this tree monochromatic with colour $d_1$, so we have $m_{T_l}(T_l,\omega|_{T_l},d_1) \leq |S_l|$. 
 
Observe also that, as $j_1, j_2 < i$, the inductive hypothesis implies that, for $l \in \{1,2\}$, \begin{align*}
f(W_l \cup \{x_l\}, \omega, d_1, j_l) & \leq \min_{T' \in \mathcal{T}_{j_l}(W_l \cup \{x_l\},G)} m_{T'}(T',\omega|_{T'},d_1) \\
									  & \leq m_{T_l}(T_l,\omega|_{T_l},d_1).
\end{align*}

Hence we see that
\begin{align*}
f(W,\omega,d_1,i)     & \leq f_1(W,\omega,d_1,i) \\
					  &	\leq f(W_1 \cup \{x_1\}, \omega, d_1,j_1) + f(W_2 \cup \{x_2\}, \omega, d_1, j_2) \\
                      & \qquad \qquad \qquad \qquad \qquad \qquad \mbox{by definition} \\
                      & \leq m_{T_1}(T_1,\omega|_{T_1},d_1) + m_{T_2}(T_2,\omega|_{T_2},d_1) \\
                      & \leq |S_1| + |S_2| \\
                      & = |S| \\
                      & = m_T(T,\omega|_T,d_1) \\
                      & = \min_{T' \in \mathcal{T}_i(W,G)} m_{T'}(T',\omega|_{T'},d_1) \\
                      & \qquad \qquad \qquad \qquad \qquad \qquad \mbox{by choice of $T$}.
\end{align*}

We now consider the case in which $T$ is monochromatic in some colour $d_2 \in C$ before the final move of $S$, so $\alpha$ simply changes the colour to $d_1$.  Note that $T$ cannot be monochromatic before the penultimate move of $S$, otherwise we could obtain a shorter sequence to flood $T$ with colour $d_1$.  Set $S'$ to be the initial segment of $S$ with just the final move omitted.  Then $S'$ must be an optimal sequence to flood $T$ with colour $d_2$, and does not make $T$ monochromatic before the final move, so we can apply the reasoning above to see that 
$$f_1(W,\omega,d_2,i) \leq |S'| = m_T(T,\omega|_T,d_1) - 1.$$
But then
\begin{align*}
f(W,\omega,d_1,i) & \leq f_2(W,\omega,d_2,i) \\
				  & \leq 1 + f_1(W,\omega,d_2,i) \\
                   & \qquad \qquad \qquad \mbox{by definition} \\
                     & \leq m_T(T,\omega|_T,d_1) \\
                     & = \min_{T' \in \mathcal{T}_i(W,G)} m_{T'}(T',\omega|_{T'},d_1) \\
                     & \qquad \qquad \qquad \mbox{by choice of $T$},
\end{align*}
completing the proof that $f(W,\omega,d_1,i) \leq \min_{T \in \mathcal{T}_i(W,G)} m_T(T,\omega|_T,d_1)$, for every $W$ and $d_1$.

To prove the correctness of the algorithm, it remains to show that for any $W \subset V$ (containing at most $k$ vertices) and $d_1 \in C$ we also have $f(W,\omega,d_1,i) \geq \min_{T \in \mathcal{T}_i(W,G)} m_T(T,\omega|_T,d_1)$.  Once again, we proceed by induction on $i$, and note that the base case for $i = 1$ holds.  Assume that $i > 1$ and that the result holds for smaller values.

In fact we prove the following claim.
\begin{claim}
For any $(W_1,W_2) \in \partition(W)$, $d_1 \in C$, $x_1x_2 \in E$ with $x_1 \notin W_2$ and $x_2 \notin W_1$, and $j_1,j_2 > 0$ such that $j_1+j_2=i$, we have
\begin{align*}
\min_{T \in \mathcal{T}_i(W,G)} m_T & (T,\omega|_T,d_1) \leq  \nonumber \\
                                    & f(W_1 \cup \{x_1\},\omega,d_1,j_1) + f(W_2 \cup \{x_2\},\omega,d_1,j_2).
\end{align*}
\end{claim}
To see that it is sufficient to prove this claim, first observe that the claim implies immediately that
\begin{equation}
\min_{T \in \mathcal{T}_i(W,G)} m_T(T,\omega|_t,d_1) \leq f_1(W,\omega,d_1,i).
\label{T<=f1}
\end{equation}
Observe also that
\begin{align*}
\min_{T \in \mathcal{T}_i(W,G)} m_T(T,\omega|_T,d_1) & \leq \min_{T \in \mathcal{T}_{i-1}(W,G)} m_T(T,\omega|_T,d_1) \\
   & \qquad \qquad \qquad \qquad \mbox{as $\mathcal{T}_{i-1}(W,G) \subseteq \mathcal{T}_i(W,G)$} \\
     											     & \leq f(W,\omega,d_1,i-1) \\
     											     & \qquad \qquad \qquad \qquad \mbox{by inductive hypothesis.}
\end{align*}     											     
Moreover, it is clear that for any $d_2 \in C$, 
$$\min_{T \in \mathcal{T}_i(W,G)} m_T(T,\omega|_T,d_1) \leq 1 + \min_{T \in \mathcal{T}_i(W,G)} m_T(T,\omega|_T,d_2)$$ 
and so, if the claim holds, it follows from (\ref{T<=f1}) that
$$\min_{T \in \mathcal{T}_i(W,G)} m_T(T,\omega|_T,d_1) \leq f_2(W,\omega,d_1,i).$$
Thus, if the claim holds, it follows that $\min_{T \in \mathcal{T}_i(W,G)} m_T(T,\omega|_T,d_1)$ is less than or equal to every expression on the right hand side of (\ref{k-linking-recursion}), giving $f(Q,\omega,d_1,i) \geq \min_{T \in \mathcal{T}_i(W,G)} m_T(T,\omega,d_1)$, as required.  Hence it is indeed sufficient to prove the claim.

We now prove the claim.  Suppose $(W_1,W_2) \in \partition(W)$, $d_1 \in C$, $x_1x_2 \in E$ with $x_1 \notin W_2$ and $x_2 \notin W_1$, and $j_1,j_2 > 0$ such that $j_1+j_2=i$.  For $l \in \{1,2\}$, pick $T_l \in \mathcal{T}_{j_l}(W_l \cup \{x_l\}, G)$ such that 
$$m_{T_l}(T_l,\omega|_{T_l},d_1) = \min_{T' \in \mathcal{T}_{j_l}(W_l \cup \{x_l\}, G)} m_{T'}(T',\omega|_{T'},d_1).$$  
Note that, by the inductive hypothesis (as $j_l < i$), we then have 
$$m_{T_l}(T_l,\omega|_{T_l},d_1) \leq f(W_l \cup \{x_l\},\omega,d_1,j_l).$$
Now set $H = T_1 \cup T_2 \cup \{x_1x_2\}$, and fix a $d_1$-minimal spanning tree $T'$ of $H$ (so, by Theorem \ref{min-tree}, $m_H(H,\omega|_H,d_1) = m_{T'}(T',\omega|_{T'},d_1)$).  Note that $T' \in \mathcal{T}_i(W,G)$.  Thus we see that
\begin{align*}
\min_{T \in \mathcal{T}_i(W,G)} m_T(T,\omega|_T,d_1) & \leq m_{T'}(T',\omega|_{T'},d_1) \\
													 & = m_H(H,\omega|_H,d_1) \\
													 & \qquad \qquad \qquad \mbox{by choice of $T'$} \\
                                                     & \leq m_{T_1}(T_1,\omega|_{T_1},d_1) + m_{T_2}(T_2,\omega|_{T_2},d_1) \\
                                                     & \qquad \qquad \qquad \mbox{by Corollaries \ref{general-decomposition} and \ref{subgraph}} \\
                                                     & \leq f(W_1 \cup \{x_1\}, \omega, d_1, j_1) + f(W_2 \cup \{x_2\}, \omega, d_1, j_2),
\end{align*}
completing the proof of the claim.

It remains only to bound the time taken to compute $f(U,\omega,d,n)$.  Note that each value of $f(W,\omega,d_1,1)$ (for any $W \subset V$ of size at most $k$ and $d_1 \in C$) can be computed in constant time.  

Suppose we have computed the value of $f(W,\omega,d_1,i)$ for each $d_1 \in C$ and $W \subset V$ of size at most $k$.  To compute $f_1(W',\omega,d_2,i+1)$ for any $W'$ and $d_2$, we take the minimum over at most $2^k$ ways to partition a set of up to $k$ points, the $|E|$ edges in the graph, the $|C|$ colours in the initial colouring, and the $2(i-1)$ ordered pairs of positive integers that sum to $i$.  Thus we take the minimum over a set of $O(2^k \cdot |E| \cdot |C| \cdot n)$ values, each of which can be computed in time $O(n)$ by adding a pair of existing values in the table, and so compute $f_1(W',\omega,d_2,i+1)$ in time $O(2^k \cdot |E| \cdot |C| \cdot i \cdot n) = O(2^k \cdot |E| \cdot |C| \cdot n^2)$.  

Once we have computed the value of $f_1$ for all entries with index $i+1$, we can compute $f_2$ for each such entry in time $O(|C|)$.  Given the values of $f_1$ and $f_2$ for each entry with index $i+1$, and the values of $f$ for entries with index $i$, we can compute $f$ for any entry with index $i+1$ in constant time.

Thus in total we require time at most $O(2^k \cdot |E| \cdot |C| \cdot n^2)$ to compute the value of $f$ for each entry in the table.  In total, the table contains $O(n^{k+1} \cdot |C|)$ entries (as there are $O(n^k)$ subsets of size at most $k$, a choice of $|C|$ colours, and $i$ takes integer values in the range $[1,n]$), so we can compute all entries, and hence determine $f(U,\omega,d,n)$, in time $O(n^{k+3} \cdot |E| \cdot |C|^2 \cdot 2^k)$.
\end{proof}

\section{Conclusions and Open Problems}

We have shown that, for any connected graph $G$, the minimum number of Flood-It moves required to make $G$ monochromatic in colour $d$ is equal to the minimum, taken over all spanning trees $T$ of $G$, of the number of moves required to flood $T$ with colour $d$. 

Using this result, we saw that \textsc{Free-Flood-It}, and the fixed variant, are solvable in polynomial time on graphs with only a polynomial number of connected subgraphs.  This proves a conjecture of Meeks and Scott \cite{general}: \textsc{Free-Flood-It} is solvable in polynomial time on subdivisions of any fixed graph. This in turn implies that \textsc{Free-Flood-It} is polynomially solvable on trees with bounded degree and a bounded number of vertices of degree at least three, although the problem is known to be NP-hard on arbitrary trees.  It would be interesting to investigate other classes of trees on which the problem can be solved in polynomial time.

Finally, we applied the result on spanning trees to $k$-\textsc{Linking-Flood-It}, demonstrating an algorithm to solve the problem in time $n^{O(k)}$.  There is potential for further investigation of the parameterised complexity of this problem, with parameter $k$: can $k$-\textsc{Linking-Flood-It} be shown to be W[1]-hard, or is there another approach to the problem which might yield a fixed-parameter algorithm?  Such an investigation could also consider a ``fixed'' variant of $k$-\textsc{Linking-Flood-It}, in which all moves must be played at some fixed vertex.


\begin{thebibliography}{9}

\bibitem{flash}
   \emph{Flood It Game}, http://floodit.appspot.com.

\bibitem{iphoneapp}
   \emph{Flood It! 2}, available at http://itunes.apple.com.

\bibitem{androidapp}
   \emph{Flood It!}, available at https://market.android.com.
   
\bibitem{madvirus}
   \emph{Mad Virus}, http://www.bubblebox.com/play/puzzle/539.htm.


\bibitem{arthurFUN}
   D.~Arthur, R.~Clifford, M.~Jalsenius, A.~Montanaro, and B.~Sach,
   The Complexity of Flood Filling Games,
   in Paolo Boldi and Luisa Gargano, editors,
   \emph{FUN}, volume 6099 of \emph{Lecture Notes in Computer Science},
   Springer,
   ISBN 978-3-642-13121-9,
   2010, pages 307-318.
      
\bibitem{honeybee}
  A.~Born, Flash application for the computer game \emph{“Biene” (Honey-Bee)}, 2009.
  http://www.ursulinen.asn-graz.ac.at/Bugs/htm/games/biene.htm.
   
\bibitem{clifford}
   R.~Clifford, M.~Jalsenius, A.~Montanaro, and B.~Sach,
   The Complexity of Flood Filling Games,
   \emph{Theory of Computing Systems} \textbf{50} (2012),
   72--92.

    
\bibitem{clr90}
   T.~Cormen, C.~Leiserson and R.~Rivest,
   \emph{Introduction to Algorithms},
   MIT Press and McGraw-Hill,
   1990.   

\bibitem{fleischer10}
   R.~Fleischer and G.~Woeginger,
   An Algorithmic Analysis of the Honey-Bee Game,
   \emph{Theoretical Computer Science} \textbf{452} (2012),
   75--87.



\bibitem{fukui}
   H.~Fukui, A.~Nakanishi, R.~Uehara, T.~Uno, Y.~Uno,
   The complexity of free flooding games,
   Information Processing Society of Jamap (IPSG) SIG Notes 2011 (August 2011),
   1-5.

\bibitem{lagoutte}
   A.~Lagoutte,
   Jeux d'inondation dans les graphes,
   Technical report,
   ENS Lyon,
   HAL: hal-00509488,
   August 2010.
   
\bibitem{lagoutte11}
   A.~Lagoutte, M.~Noual, E.~Thierry,
   Flooding games on graphs,
   HAL: hal-00653714,
   December 2011.

\bibitem{general}
  K.~Meeks and A.~Scott,
  The complexity of flood-filling games on graphs,
  \emph{Discrete Applied Mathematics} \textbf{160} (2012),
  959--969.
  
  
\bibitem{2xn}
   K.~Meeks and A.~Scott, 
   The complexity of Free-Flood-It on $2 \times n$ boards, 
   arxiv.1101.5518v1 [cs.DS],
   January 2011.





\end{thebibliography}
\end{document}